\documentclass[10pt,conference,letterpaper]{IEEEtran}
\IEEEoverridecommandlockouts

\usepackage{times,epsfig,amssymb,graphicx,amsfonts,amsthm}
\usepackage{empheq}
\usepackage{graphicx}
\usepackage{mwe}
\usepackage{flushend}
\usepackage{psfrag}
\usepackage{fge}
\usepackage{amsthm}
\usepackage{amsfonts}   
\usepackage{amssymb}    
\usepackage{mathrsfs}
\usepackage{caption}
\usepackage{setspace}
\usepackage{tikz}
\usepackage{mathtools}
\usepackage{mathdots}
\usepackage{amsmath}
\usepackage{relsize}
\mathtoolsset{showonlyrefs}
\usepackage{dblfloatfix} 
\usepackage{caption}
\usepackage{subcaption}
\usepackage{tkz-tab}
\usepackage{color,soul}
\usepackage{algorithm}
\usepackage{algcompatible}
\usepackage{color}
\usepackage{mathrsfs}
\floatname{algorithm}{Algorithm}
\definecolor{yellow}{rgb}{.90,.95,1}
\sethlcolor{yellow}

\newtheorem{mydef}{Definition}
\newtheorem{assumption}[mydef]{Assumption}
\newtheorem{mytheorem}{Theorem}
\newtheorem{mylemma}[mytheorem]{Lemma}

\newtheorem{myremark}{Remark}

\newcounter{ale}

\newenvironment{liste}{\begin{itemize}}{\end{itemize}}
\newcommand{\aliste}{\begin{liste} \setcounter{ale}{1}}
\newcommand{\zliste}{\end{liste}}

\title{\Large \bf  Controllability and Fraction of Leaders in Infinite Networks }
\author{C. Enyioha, M. A. Rahimian, G. J. Pappas and A. Jadbabaie{\small$~^{\dagger}$}
\thanks{$^{\dagger}$ All authors are with the Department of Electrical and Systems Engineering, University of Pennsylvania, Philadelphia, PA 19104-6228 USA (email: {\fontsize{8}{8}\selectfont\ttfamily\upshape \{cenyioha, mohar, pappasg jadbabai\}@seas.upenn.edu}).} 
\thanks{This work was supported in part by TerraSwarm, one of six centers of STARnet, a Semiconductor Research Corporation program sponsored by MARCO and DARPA, and in part by AFOSR Complex Networks Program.}
}

\begin{document}
\maketitle

\begin{abstract} 
In this paper, we study controllability of a network of linear single-integrator agents when the network size goes to infinity. We first investigate the effect of increasing size by injecting an input at every node and requiring that network controllability Gramian remain well-conditioned with the increasing dimension. We provide theoretical justification to the intuition that high degree nodes pose a challenge to network controllability. In particular, the controllability Gramian for the networks with bounded maximum degrees is shown to remain well-conditioned even as the network size goes to infinity. In the canonical cases of star, chain and ring networks, we also provide closed-form expressions which bound the condition number of the controllability Gramian in terms of the network size. We next consider the effect of the choice and number of leader nodes by actuating only a subset of nodes and considering the least eigenvalue of the Gramian as the network size increases. Accordingly, while a directed star topology can never be made controllable for all sizes by injecting an input just at a fraction $f<1$ of nodes; for path or cycle networks, the designer can actuate a non-zero fraction of nodes and spread them throughout the network in such way that the least eigenvalue of the Gramians remain bounded away from zero with the increasing size. The results offer interesting insights on the challenges of control in large networks and with high-degree nodes.
\end{abstract}

\section{Introduction \& Background}

The literature on the control of networks is vast and continues to attract much attention amongst diverse communities ranging from controls and theoretical physics to biology and applied sciences. In \cite{lombardi2007controllability} for instance, an interpretation of the controllability matrix is presented and applied to networks in biology for monitoring protein concentrations; while in \cite{gu2014controllability}, controllability of Brain networks is investigated.

In the control community as well, Pasqualetti \textit{et al.} in \cite{pasqualetti2013controllability} study the problem of controlling complex networks and quantify the difficulty of the control problem as a function of the minimum energy control. There, they also derived bounds to analyze the trade-off between control energy and number of driver nodes. Whereas earlier works started by \cite{TannerDistinctEigvOrthEigv} and later carried through by Mesbahi, Egerstedt and their collaborators \cite{rahmani2009controllability,martini2010controllability} have been focused on Laplacian dynamics, where leader nodes update their state values based on exogenous inputs and non-leader nodes update their states according to their relative states with their neighbors. Existing literature on controllability of networks has mostly focused on undirected networks.

In this paper, we consider the problem of controllability for a directed or undirected network of linear single-integrator agents and investigate the core challenges of control as network size increases. To begin, we assume that each agent is injected with an exogenous control signal and there our primary contribution is in bounding the condition number of the controllability Gramian in terms of the singular values of the network matrix, such that the Gramian remains numerically stable with the increasing dimension. In particular, we show that in structures with a bounded maximum degree the controllability Gramian remains well-conditioned even as the network size increases. Controllability of large networks and the interplay between structure and degree distribution has been a focus of recent studies \cite{Miao20086225,correlations}. Our results supplement the existing literature by providing the Gramian condition number as a metric to test controllability with the increasing network size; hence, highlighting the challenges posed by the high degree nodes on the network controllability. We next shift attention to the choice of leaders, i.e. exogenously actuated nodes, in the canonical cases of star, path and cycle networks and point out their main difference with respect to the spectral radius of the controllability Gramian inverse. In particular, while the star network can never be made controllable for all sizes just by selecting a fraction $f<1$ of nodes as leaders, in cases of the path and ring networks, one can select a non-zero fraction of nodes and spread them across the network to maintain controllability with the increasing size. 

The rest of this paper is organized as follows. The model and problem formulation are presented in Section \ref{sec:pre}. In Section~\ref{sec:main-result}, we present our main result on the numerical stability of the Gramian with the increasing dimension, and follow up with illustrations on canonical networks. In section \ref{sec:ratio-of-leaders}, we investigate the effect of the ratio and location of designated leader nodes on the controllability properties of star, path and cycle networks and with the increasing sizes. Concluding remarks are provided in Section~\ref{sec:conc}. 

\section{Preliminaries}\label{sec:pre}

\subsection{Network Information Flow Graph}

Throughout the paper, $\mathbb{R}$ is the set of all real numbers, $\mathbb{N}$ is the set of all natural numbers, $N\in\mathbb{N}$ denotes the network size, and $\mathcal{N} = \{1,\ldots,N\}$. Matrices are represented by capital letters, vectors are expressed by boldface lower-case letters, and the superscript $^{T}$ indicates the matrix transpose. Moreover, for a matrix $D$, $\left[D\right]_{ij}$ indicates the element of $D$ which is located at its $i-$th row and $j-$th column, and D is symmetric if $D = D^{T}$. We denote as $\mathcal{G} = (\mathcal{N},\mathcal{E})$ a (directed or undirected) graph comprising $N$ nodes labeled by $\mathcal{N}$, and $\mathcal{E} \subset \mathcal{N}\times \mathcal{N}$ the set of edges of $\mathcal{G}$. Agents $i$ and $j$ are called neighbors if $(i,j)\in \mathcal{E}$, graphs are used to capture the network information flow structure and we say that $(i,j)$ is an edge from $i$ to $j$, and represent it by an arrow starting from $i$ and ending at $j$. Given $\mathcal{G}$, we denote the network (weighted adjacency) matrix of the graph $\mathcal{G}$ by $A \in \mathbb{R}^{N \times N}$, where the entries of $A$ are such that $[A]_{ji} = 0$ if edge $(i,j)\not\in \mathcal{E}$. $A$ is a symmetric matrix iff the graph is undirected (symmetric). The eigenvalues of the matrix $A$ are denoted by $\lambda_1(A) \geq \lambda_2(A) \geq \hdots \geq \lambda_N(A)$, and its the singular values are denoted by $\sigma_1(A) \geq \sigma_2(A) \geq \hdots \geq  \sigma_N(A)$ and given as $\{\sigma^2_i(A), i \in \mathcal{N}\} = \{\lambda_i(AA^{T}), i \in \mathcal{N}\}$. An infinite network is a network $\mathcal{G} = (\mathcal{N},\mathcal{E})$, in which $\mathcal{N}$ is countably infinite so that $\mathcal{N} \leftrightarrow \mathbb{N}$. A locally $M$-bounded network is a network $\mathcal{G} = (\mathcal{N},\mathcal{E})$ together with its associated matrix $A$, satisfying $\forall j, \sum_{i\in \mathcal{N}} |a_{ji}| < M$, and $\forall i, \sum_{j\in \mathcal{N}} |a_{ij}| < M$, where $M < \infty$ is a bounded constant.

\subsection{The Model}

We consider a network of $N$ single integrator agents, which are labeled from $1$ to $N$ and whose interaction structure is expressed by the graph $\mathcal{G}$. We assume discrete-time dynamics in the interaction of the networked agents and let $x_i$, $i \in \mathcal{N}$ represent the scalar state of agent $i$ such that the temporal evolution of the agents after a fixed initial time $t_0 \in \mathbb{N}$ is given by: 
\begin{equation}\label{eq:dynamics}
\mathbf{x}(t+1) = A{\mathbf{x}}(t) + B{\mathbf{u}}(t), \ \ t > t_0, t \in \mathbb{N},
\end{equation} where $A$ is the network (or adjacency) matrix describing the interaction links between agents, $\mathbf{x}(t) = [x_1(t),x_2(t),\ldots,x_N(t)]^T$ is the state vector of the nodes, $B = I$ is the input matrix and ${\mathbf{u}}(t) \in \mathbb{R}^n$ is an exogenous control input signal injected at each node in the network.  We make the following assumptions in our modeling.

\begin{assumption}\label{assump:Stable_A}
The network matrix $A$ is Schur stable; that is, all its eigenvalues are strictly inside the unit circle.
\end{assumption}

\begin{assumption}\label{assump:b-equal-i}
The input matrix $B$ is an $N \times N$ diagonal matrix, whose diagonal entries consist only of $0$ and $1$.
\end{assumption}

\begin{myremark}\label{rem:structureOfB} Notably, Assumption \ref{assump:b-equal-i} is significant in that any such choice of matrix $B$ indicates a particular selection of leader nodes, which are those nodes to which a designer has access and can feed them with control signals. The diagonal structure of $B$ further implies that the leader nodes are driven independently of each other. In particular setting $B = I$, to imply that the exogenous control input signals are injected at each node in the network, allows us to investigate the controllability properties of the network as reflected through the solution of \eqref{eq:L} and solely determined by the network matrix $A$. As we shall see, this feature plays a key role in helping us characterize the influence of network size and maximum degree on controllability, and is distinct from much of the existing literature where the notion of driver nodes are typically considered \cite{pasqualetti2013controllability,liu2011controllability,rahimian2013structural}. \end{myremark}

\begin{assumption}\label{assump:localboundedness}
The network matrix $A$ is locally $M$-bounded.
\end{assumption}

\begin{myremark}\label{rem:stab} It is worth highlighting that stability and controllability properties of $A$ differ in the sense that while stability of $A$ is sensitive to perturbations in the network matrix $A$, controllability is not. Rather controllability is sensitive to structural changes. As such, if the given network matrix $A$ in \eqref{eq:dynamics} is not stable, it is possible to shift its eigenvalues to make it Schur stable, by scaling its entries so that they lie within the unit circle and without affecting its controllability property. We shall make use of this feature when considering a family of networks with a particular structure but of varying sizes, as we can ensure that all network matrices are Schur stable by uniformly scaling all members of the family by some large enough constant $\gamma$. \end{myremark}

\subsection{Network Controllability Gramian, its Condition Number, and Relation to Minimum Energy Control}

The networked system in \eqref{eq:dynamics} is controllable if any state $\mathbf{x}(t_0)$ can be steered to the zero state $\mathbf{0} = \mathbf{x}(t_1)$, for some finite $t_1 > t_0$ and using an appropriate input signal $\mathbf{u}(t), t_0 \leq t \leq t_1$. This  controllability condition for a Schur stable matrix $A$ is equivalent to requiring that the solution to the discrete Lyapunov equation 
\begin{equation}\label{eq:L}
AG_cA^T - G_c = -BB^{T}
\end{equation} is  invertible. The controllability Gramian is the symmetric positive semi-definite matrix $G_c$ that uniquely satisfies \eqref{eq:L} and is given by \cite[Chapter 6]{Chen:1998:LST:521603},
\begin{equation}\label{eq:G}
G_c = \lim_{t\to\infty}G(t)\mbox{, where, } G(t) = \sum_{\tau=0}^t \ A^\tau {B} {B^T} (A^T)^\tau.
\end{equation} The controllability condition is equivalent to positive-definiteness of $G_c$. The difficulty of control can be quantified by the minimum amount of energy required to reach a state $\mathbf{x}(t) = \mathbf{x}_{des}$ from $\mathbf{x}(0) = \mathbf{0}$, which is equal to $\mathbf{x}_{des}^{T}G^{-1}(t)\mathbf{x}_{des}$ and can be achieved through the least norm input $\mathbf{u}(\tau) =  \mathbf{u}^{*}(\tau)$ given by $\mathbf{u}^{*}(\tau) = B^{T}(A^{T})^{t-1-\tau}G^{-1}(t-1)\mathbf{x}_{des}$ for all $\tau \in [t-1]$. However, for $A$ Schur stable per Assumption~\ref{assump:Stable_A}, $G(t)$ converges to $G_c$, exponentially fast and for sufficiently large $t$, the two matrices can be made arbitrarily close. In particular, if $G_c$ is nearly singular, then large energy inputs are required to reach those states $\mathbf{x}_{des}$ belonging to the eigenspace of its least eigenvalue $\lambda_N(G_c)$. This motivates the use of the minimum eigenvalue of the controllability Gramian in \cite{pasqualetti2013controllability}, and we adopt the same measure of the worst case control effort when investigating the role of the choice and fraction of leader nodes in Section~\ref{sec:ratio-of-leaders}.

Moreover, when investigating the problem of network controllability with the increasing size, it becomes crucial for large $N$ that computations of $G^{-1}(t)$ for minimum energy control remain numerically stable; that is, for $G^{-1}(t)$ to be well-conditioned as the dimension $N$ increases \cite[Chapter III]{trefethen1997numerical}. To this end, we require that the Gramian condition number, $\kappa(G_c) \triangleq {\sigma_{\text{max}}(G_c)}/{\sigma_{\text{min}}(G_c)}$, with $\sigma_{\text{max}}(G_c)$ and $\sigma_{\text{min}}(G_c)$ being the maximal and minimal singular values of $G_c$, remain bounded uniformly in $N$. This is especially important when we guarantee that ${\sigma_{\text{min}}(G_c)}$ is bounded away from zero by taking $B=I$, as then even though the network is controllable for any finite $N$, for certain networks as $N\rightarrow \infty$ the Gramin condition number grows unbounded. Examples of such networks are star and complete networks, as shown in Fig. \ref{fig:star+complete}. 
\begin{figure}[h*]
\centering
\includegraphics[width=0.9\linewidth]{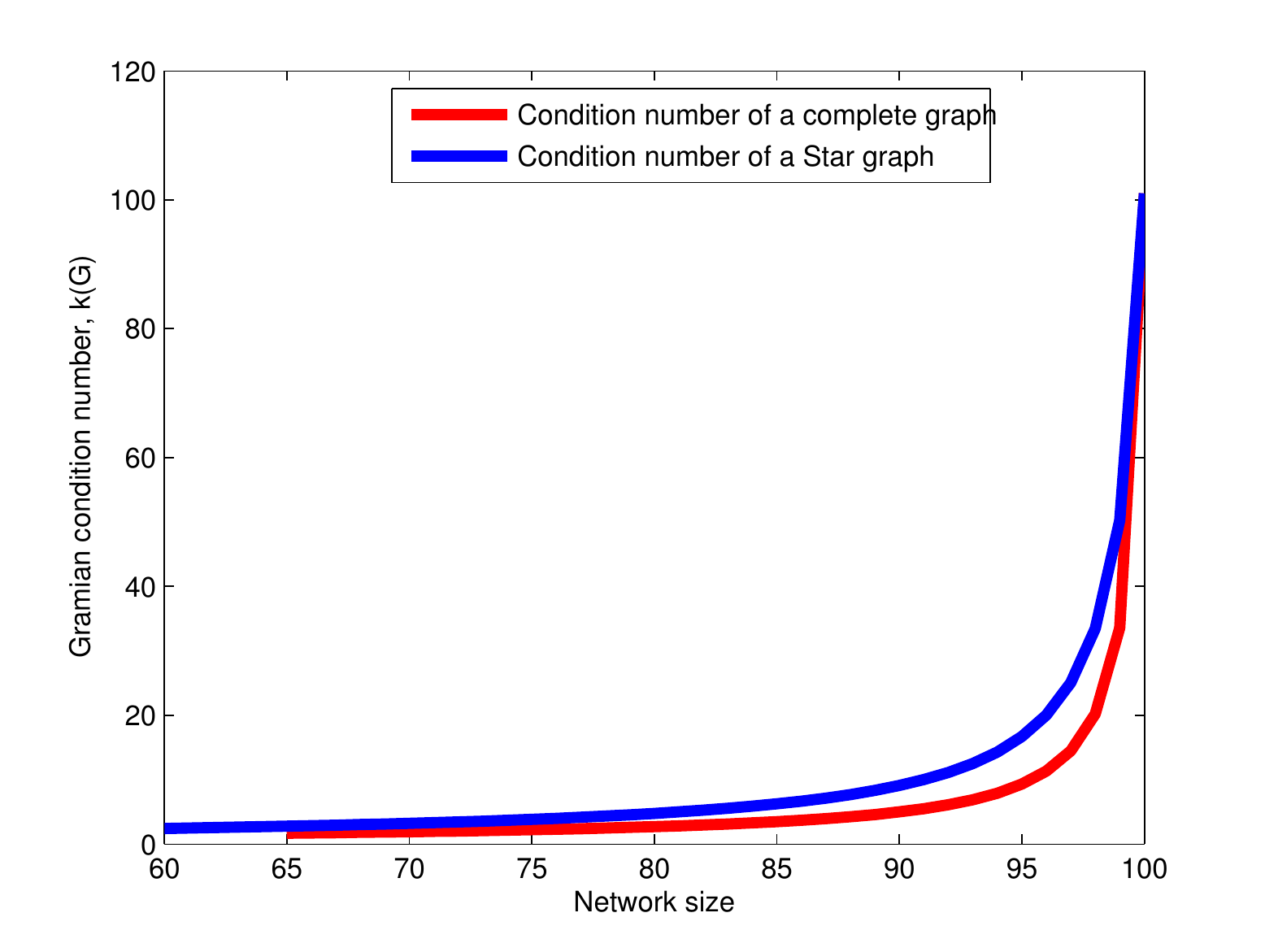}
\caption{The plot above depicts how the condition numbers $\kappa(G_c)$ of undirected star and complete networks grow unbounded as the network size $N\rightarrow \infty$.}
\vspace*{-9pt}
\label{fig:star+complete}
\end{figure}
Clearly, the controllability Gramian for certain networks becomes ill-conditioned as the network size increases. The importance of condition number for the controllability Gramian and the network control problem is also highlighted in \cite{plos,PhysRevLett.110.208701} and our main results in the next section provide a sufficient characterization of the networks for which $\kappa(G_c)$ remains bounded as $N\rightarrow \infty$.

\section{Controllability \& Bounded Degrees}\label{sec:main-result}

As a key insight, in this section we characterize how the increasing local degrees in a network hinders its controllability property. The main result of this section provides a theoretical justification to this intuition, resulting in a sufficient condition for having a well-conditioned Gramian as network size increases. First, we state a lemma bounding the singular values of the Controllability Gramian which we use in the sequel. 

\begin{mylemma}[\cite{gajic2008lyapunov}]
Let the matrix $A$ in \eqref{eq:dynamics} be asymptotically stable such that the solution $G_c=G_c^T\succ 0$ to \eqref{eq:L} exists. Furthermore, let $\alpha_1 \geq \hdots \geq \alpha_N$ be the eigenvalues of $G_c$, 
$\beta_1 \geq \hdots \geq \beta_N$ be the eigenvalues of $B$,  $\text{Re}(\lambda_1) \geq \hdots \geq \text{Re}(\lambda_N)$ be the eigenvalues of $A$, and $1 > \sigma_1^2 \geq \hdots \geq \sigma^2_N$ be the eigenvalues of $AA^T$. Then, the eigenvalues of $G_c$ are upper and lower bounded by $\beta_i + \frac{\sigma^2_N \beta_N}{1-\sigma^2_N} \leq \alpha_i \leq \beta_i + \frac{\beta_1 \sigma^2_1}{1-\sigma^2_1}$, $\forall i \in \mathcal{N}$.
\label{lem:gramian_bounds}
\end{mylemma}
\begin{proof}
The proof is a direct consequence of the Ostrosky inequalities for the eigenvalue of a sum of symmetric matrices and for the eigenvalue of a matrix product. We refer readers to \cite[Theorem 3.1]{gajic2008lyapunov} for proof of Lemma \ref{lem:gramian_bounds}. Similar and related results are presented in  \cite{karanam1981lower} and \cite{mori1986explicit}.
\end{proof}

Observe that since $G_c = G_c^T$, its singular values and eigenvalues coincide and Lemma~\ref{lem:gramian_bounds} can be used to bound the condition number of controllability Gramian for the network model given by \eqref{eq:dynamics} under the Assumptions \ref{assump:Stable_A} to \ref{assump:localboundedness}.

\begin{mytheorem}\label{th:main-result} Given $B=I$, together with Assumptions 1 and 3 for the network model \eqref{eq:dynamics}, the condition number of the controllability Gramian $\kappa(G_c)$ is bounded in terms of the singular values of the network matrix $A$, as follows:
\begin{equation}\label{eq:cond_bound}
\kappa(G_c) = \frac{\sigma_1(G_c)}{\sigma_N(G_c)} \leq \frac{1-\sigma^2_N(A)}{1-\sigma^2_1(A)}.
\end{equation}
\end{mytheorem}

\begin{proof} 
First since $B=I$, $\beta_i = 1, \forall i$ and the spectral bounds of $G_c$ from Lemma \ref{lem:gramian_bounds} become
\begin{equation}\label{eq:gramian_bounds2}
\frac{1}{1-\sigma^2_N} \leq \alpha_i \leq \frac{1}{1-\sigma^2_1},\forall i \in \mathcal{N}.
\end{equation}
 We can now upper bound $\kappa(G_c)$ as in \eqref{eq:cond_bound}, noting that $\kappa(G_c)$ is trivially lower bounded by $1$.
\end{proof}

The bounds in \eqref{eq:cond_bound} are in terms of the singular values of the network matrix $A$, and the following result attributed to \textit{Schur} allows us to uniformly bound $\sigma_{i}^2(A),\forall i \in \mathcal{N}$ of an adjacency matrix $A$, provided that its maximum degree does not scale with the network size $N$.

\begin{mylemma}(Schur Bound \cite{schur}) \label{lem:schur}
Let $A$ be an $N\times N$ locally $M$-bounded network matrix; then its largest singular value satisfies $\sigma_{1}^2(A) \leq M^2$.
\end{mylemma}

\begin{proof} For all $i,j \in \mathcal{N}$, let $R_i = \sum_{k \in \mathcal{N}}|[A]_{ik}|$ and $C_j = \sum_{k \in \mathcal{N}} |[A]_{kj}|$. It follows by the Schur Bound \cite{schur,Golub:1996:MC:248979}, that $\sigma^2(A)$ $ \leq $ $\max_{i\in  \mathcal{N}} \sum_{j\in  \mathcal{N}} |[A]_{ij}| C_j$ $ \leq$ $ \max_{i,j \in \mathcal{N},[A]_{ij}\neq 0 } R_i C_j$, and by locally $M$-boundedness we get that $R_i<M$ and $C_j<M$, $\forall i,j\in \mathcal{N}$, so that the claimed bound follows. \end{proof}

Lemma \ref{lem:schur} implies that for locally bounded networks and after a proper normalization to ensure it is Schur stable, we can derive upper bounds for $\kappa(G_c)$ that does not scale with the network size $N$ and hence ensure controllability as $N \to \infty$. This leads us to our main result on controllability of locally bounded infinite networks.

\begin{mytheorem} \label{th:LocalBounded} Let $A$ be the network matrix corresponding to a locally $M$-bounded network and $\gamma > M$ constant. The condition number of the Gramian for a network, following the dynamics in \eqref{eq:dynamics} with network matrix $\frac{1}{\gamma}A$ and the input matrix $B=I$, is bounded uniformly in $N$, whence the Controllability Gramian is guaranteed to remain well-conditioned as $N \to \infty$. \end{mytheorem}

\begin{proof} It follows from Lemma~\ref{lem:schur} that the singular values of $\frac{1}{\gamma}A$ are bounded above by $\frac{M}{\gamma} <1$. Replacing the latter inequality in \eqref{eq:cond_bound} and noting that $\frac{1}{\gamma}A$ is Schur stable we get
\begin{align}
\kappa(G_c) \leq \frac{1-\sigma^2_N(A)}{1-\sigma^2_1(A)}  \leq \frac{\gamma^2}{\gamma^2-{M}^2} <  \infty. 
\label{eq:gramian_bounds}
\end{align}
\end{proof}

\begin{myremark}The result of Theorem \ref{th:LocalBounded} is to a great extent an artifact of our methodology. In particular, by taking $B=I$ the minimum eigenvalue $\alpha_N$ of $G_c$ is lower-bounded by one and away from zero per \eqref{eq:gramian_bounds2}. Indeed, setting $B=I$ and allowing for an input signal to be injected at every node of the network factors out the variety of structural and dynamical influences that affect the control behavior, whence singling out the effect of network size $N$. This in turn enables us to highlight the role of maximum degree, or more generally local boundedness, in controllability of large networks. Our result shows that though for each finite $N$ the network is controllable, as $N$ goes to infinity being locally-bounded is a sufficient condition for the controllability Gramian to remain well-conditioned. \end{myremark}

By considering the condition number of the controllability Gramian, we are able to use bounds on $\kappa(G_c)$ to investigate the effect of network size $N$, and the limiting behavior as $N\to \infty$. This idea is explored further in the next subsection, where we consider the cases of star, path and cycle networks and proffer closed form expressions for the upper-bound in \eqref{eq:cond_bound}.

\subsection{Bounds on Condition Numbers for Canonical Networks}

In this subsection, we illustrate our key result on some canonical graphs. For the cases considered, we compute bounds on $\kappa(G_c)$ and consider the limit as $N\rightarrow \infty$ of $\kappa(G_c)$. In each case, based on the premise of Theorem \ref{th:LocalBounded} and per Remark \ref{rem:stab}, we scale the $0-1$ adjacency matrices by a common constant $\gamma$ to ensure the Schur stability of $\frac{1}{\gamma}A$ for every network in the range of sizes considered.
 
\textbf{Undirected star} graphs on $N$ nodes have eigenvalues that are given by $\lambda_i = 0$, $\forall i \in\{2, \ldots, N-1 \}$, and $\lambda_N = - \sqrt{N-1}, \lambda_1 = + \sqrt{N-1}$ \cite{yuan2013exact}. Based on \eqref{eq:gramian_bounds} in the proof of Theorem \ref{th:LocalBounded}, we can bound $\kappa(G_c)$ as follows: $\kappa(G_c) \leq  \frac{\gamma^2}{\gamma^2 - (N-1)}$. First, we note that star networks do not satisfy the premise of Theorem \ref{th:LocalBounded}, since its maximum degree is not bounded as $N\rightarrow \infty$.  As shown in Fig. \ref{fig:star+complete}, a star network is a perfect archetype of networks that become uncontrollable since its maximum degree is unbounded as $N\rightarrow \infty$, causing the condition number of its associated Gramian, $\kappa(G_c)$ to grow unbounded as $N\rightarrow \infty$.
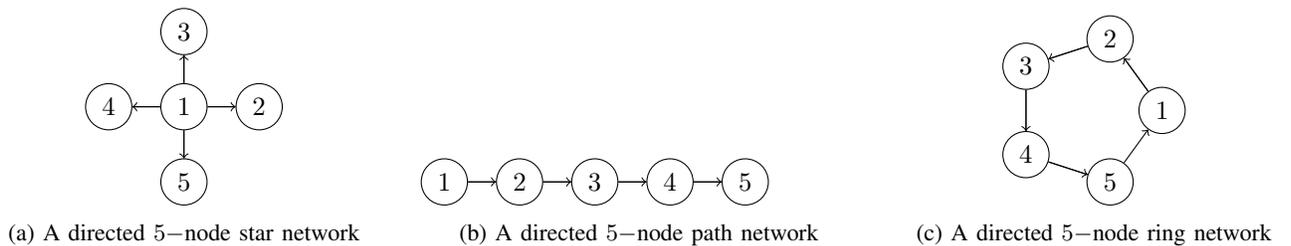
\begin{figure*}[hb]
        \centering
        \begin{subfigure}[b]{0.32\textwidth}
\centering
\begin{tikzpicture}
\tikzstyle{every node}=[draw,shape=circle];
\node (v0) at (0:0) {$1$};
\node (v1) at ( 0:1) {$2$};
\node (v2) at ( 90:1) {$3$};
\node (v3) at (2*90:1) {$4$};
\node (v4) at (3*90:1) {$5$};

\foreach \from/\to in {v0/v1, v0/v2, v0/v3, v0/v4}
\draw [->] (\from) -- (\to);

\draw
(v0) -- (v1)
(v0) -- (v2)
(v0) -- (v3)
(v0) -- (v4);

\end{tikzpicture}
\caption{A directed $5-$node star network}
\label{fig:dir-star}
\end{subfigure}~ 
\begin{subfigure}[b]{0.32\textwidth}
\begin{tikzpicture}
\tikzstyle{every node}=[draw,shape=circle];
\node (v0) at (0:0) {$1$};
\node (v1) at ( 0:1) {$2$};
\node (v2) at ( 0:2) {$3$};
\node (v3) at (0:3) {$4$};
\node (v4) at (0:4) {$5$};

\foreach \from/\to in {v0/v1, v1/v2, v2/v3, v3/v4}
\draw [->] (\from) -- (\to);

\draw
(v0) -- (v1)
(v1) -- (v2)
(v2) -- (v3)
(v3) -- (v4);

\end{tikzpicture}
\caption{A directed $5-$node path network}
\label{fig:dir-path-graph}
\end{subfigure}~ 
\begin{subfigure}[b]{0.32\textwidth}
\centering
\begin{tikzpicture}
\tikzstyle{every node}=[draw,shape=circle];
\node (v0) at ( 0:1) {$1$};
\node (v1) at ( 72:1) {$2$};
\node (v2) at (2*72:1) {$3$};
\node (v3) at (3*72:1) {$4$};
\node (v4) at (4*72:1) {$5$};

\foreach \from/\to in {v0/v1, v1/v2, v2/v3, v3/v4, v4/v0}
\draw [->] (\from) -- (\to);

\draw

(v0) -- (v1)
(v1) -- (v2)
(v2) -- (v3)
(v3) -- (v4);
(v4) -- (v0);

\end{tikzpicture}
\caption{A directed $5-$node ring network}
\label{fig:dir-ring-graph}
\end{subfigure}
\caption{Edge orientations of Directed Graphs considered}\label{fig:graphs-directed}
\end{figure*}

\vspace*{1\baselineskip} 

\textbf{Undirected path} graphs have a maximum degree of two that is constant, hence bounded, as the network size $N\rightarrow \infty$. The eigenvalues of an undirected path network with $N$ nodes are given by $\lambda_i = 2 \cos \left(\frac{i\pi}{N+1}\right)$, $\forall i \in\mathcal{N}$ \cite{yuan2013exact}. Hence, $\kappa(G_c)$ for a path network is upper bounded by 
\begin{equation}\label{pathbound}
\kappa(G_c) \leq \frac{1 - {\left(\frac{2 \cos \left(\frac{\lfloor N/2 \rfloor \pi}{N+1}\right)}{\gamma}\right)}^2}{1-{\left(\frac{2 \cos \left(\frac{\pi}{N+1}\right)}{\gamma}\right)}^2} = \frac{{\gamma}^{2} - 4 \cos^2 \left(\frac{\lfloor N/2 \rfloor \pi}{N+1}\right)}{{\gamma}^2- 4 \cos^2 \left(\frac{\pi}{N+1}\right)}.
\end{equation}
\begin{figure}[h*]
\centering
\includegraphics[width=0.9\linewidth]{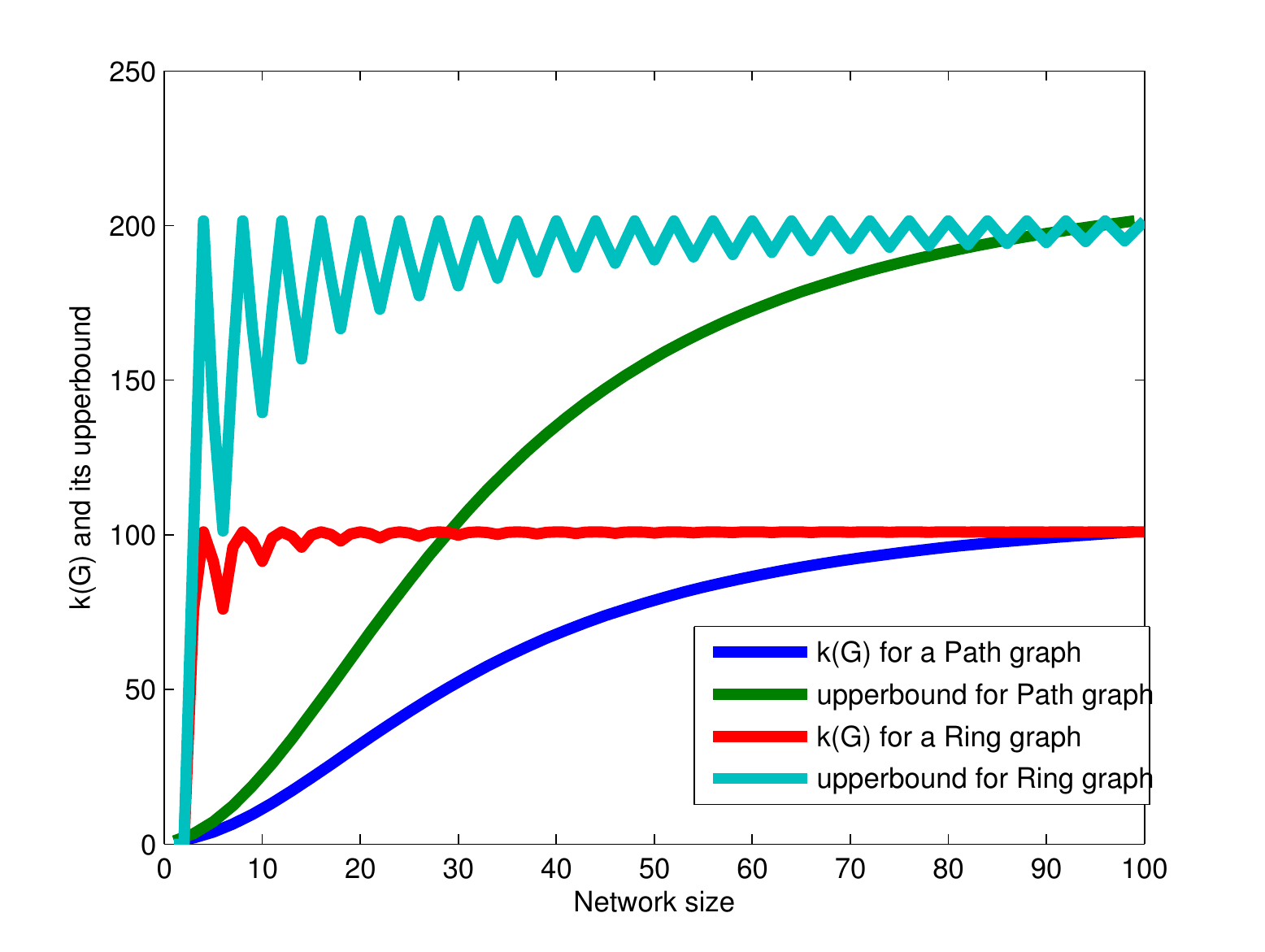}
\caption{Gramian condition numbers for ring and path graphs.}
\label{fig:ring+path-u}
\vspace*{-15pt}
\end{figure}
As $N \rightarrow \infty$, we can see that the upper bound of $\kappa(G_c)$ in undirected path graphs, as shown in Fig. \ref{fig:ring+path-u}, is bounded, and approaches its bound from below. Specifically, as $N \to \infty$, $\cos^2(\frac{\pi}{N+1}) \to 1$ and $\cos^2 \left(\frac{\lfloor N/2 \rfloor \pi}{N+1}\right) \to 0$, so that $\kappa(G_c)$ for a path network is essentially upper bounded by $\frac{\gamma^2}{\gamma^2 - 4}$.  

\textbf{Undirected ring} graphs remain locally bounded as the network size increases, similarly to undirected path graphs. The eigenvalues of a ring network of size $N$ is given by $\lambda_i = 2 \cos \left(\frac{2\pi (i-1)}{N}\right), i \in\{1, \ldots, N\}$; hence, the condition number of the controllability Grammian for a ring network is upper bounded as $\kappa(G_c) \leq \frac{\gamma^2 - 4 \cos^2 \left(\frac{\lfloor N/2 \rfloor \pi}{N}\right)}{{\gamma}^2- 4}$. The behavior of the upper bound on $\kappa(G_c)$ is similar to that of the Path graph. Shown in Fig. \ref{fig:ring+path-u}, the periodic spikes observed in the plot are due to the term $4 \cos^2 \left(\frac{\lfloor N/2 \rfloor \pi}{N}\right)$. In particular, for low values of $N$, the differences in the values of $\frac{\lfloor N/2 \rfloor \pi}{N}$ are higher; and as $N$ increases, the differences reduce, resulting in the evening out of the `saw-tooth' observed for low values of $N$; and as $N\rightarrow \infty$ the term $4 \cos^2(\frac{\lfloor N/2 \rfloor \pi}{N})$ approaches $0$. Observe that for both path and ring graphs we get the same asymptotic bound of $\frac{\gamma^2}{\gamma^2 - 4}$, also captured by Fig. \ref{fig:ring+path-u}; and indeed, it is to be expected that ring and path networks should behave increasingly similar to each other as $N \to \infty$.

\textbf{Undirected complete} graphs do not satisfy the premise of Theorem \ref{th:LocalBounded}. In particular, the eigenvalues of an undirected complete network are given by $\lambda_1 = N-1$, $\lambda_i =  -1$, $\forall i \in \{2, \ldots , N \}$. Hence, the condition number of the controllability Gramian for a complete network is upper bounded by $\kappa(G_c) \leq \frac{\gamma^2 - 1}{{\gamma}^2- (N-1)^2}$. Like star networks, complete networks are not locally bounded as $N \rightarrow \infty$. Hence, the sufficient conditions in Theorem \ref{th:LocalBounded} are not satisfied and as we observed in Fig. \ref{fig:star+complete},  $\kappa(G_c)$ for complete graphs grows unbounded with increasing network size.

\textbf{Directed star} networks have a constant condition number on the controllability Gramian, even though the bound on $\kappa(G_c)$ increases unbounded with as $N\rightarrow \infty$. This observation is intuitive, since an application of control input at the central node affects other nodes to control the network, implying that the network can be controlled with low energy. The squared singular values of a directed star networks with edge orientation as shown in Fig. \ref{fig:dir-star} are $\sigma_i^2= 0, \forall i\in \{1,\hdots, N-1\}$ and $\sigma_N^2 = N-1$. Substituting these into \eqref{eq:cond_bound}, we have that
$\kappa(G_c) \leq \frac{\gamma^2}{\gamma^2 - (N-1)^2}$, where for the range of values that $N$ takes, the sale factor $\gamma$ is such that it dominates the largest $N$, thence for the directed star the bound increases as $N$ increases. Numerical experiments indicate that the actual condition number of the Gramian associated with the directed star network is bounded, pointing out that locally-boundedness in Theorem~\ref{th:LocalBounded} is a sufficient but not necessary condition.

\textbf{Directed path} networks have maximum degree that is bounded as $N\rightarrow \infty$. For directed path graphs with edge orientation shown in Fig. \ref{fig:dir-path-graph}, the squared singular values are $\sigma_i^2= 1$, for $ i = 1,\hdots,N-1$ and $\sigma_N^2 = 0$, which yield an upper bound of $\kappa(G_c) \leq \frac{\gamma^2}{\gamma^ - 1}$, applying \eqref{eq:cond_bound}. Observe that the bound is constant; in fact, $\kappa(G_c) = \frac{\gamma^2}{\gamma^ - 1}$, $\forall \ N$ and as $N\rightarrow \infty$, in directed path networks.

 \textbf{Directed ring} networks with edge orientation as shown in Fig. \ref{fig:dir-ring-graph} have squared singular values given as $\sigma_i^2 = 1 $, $\forall i\in\mathcal{N}$. From \eqref{eq:cond_bound} we can get the bound $\kappa(G_c) \leq 1$, which is constant and in fact binding: $\kappa(G_c) = 1, \forall \ N$ and as $N\rightarrow \infty$.

\section{Controllability \& Fraction of Leaders in Directed Canonical Structures}\label{sec:ratio-of-leaders}

Thus far, in analyzing the effect of increasing size on network controllability, we have assumed that all nodes are injected with an input, so that $B=I$. In this section, we study how the fraction and spread of leaders in the three directed structures (depicted in Fig. \ref{fig:graphs-directed}), affect their controllability properties. Thence, rather than set the input matrix $B=I$, we inject the inputs only into a subset of nodes, dubbed leaders. While condition number has been effective in investigating the effects of increasing size in Section~\ref{sec:main-result}, here we adopt minimum required energy in the worst case captured by $1/\lambda_{N}(G_c)$ as the measure of interest for investigating the role of leader nodes. 

To begin, consider the cases of the star and path network in Figs. \ref{fig:dir-star} and \ref{fig:dir-path-graph} with their respective $N \times N$ adjacencies $A_s$ and $A_p$ given by 
\begin{align}
A_s = \left[\begin{array}{cccc}
0 & 0 & \cdots & 0\\
1 & 0 & \cdots & 0\\
\vdots & \vdots & \ddots & \vdots\\
1 & 0 & \cdots & 0
\end{array}\right], A_p = \left[ \begin{array}{cccccc}
    0                &  0       &   0          &    \cdots             &      0     \\
    1                &  0      &              &         \iddots        &   \vdots  \\
    0                &  \ddots  &   \ddots           &           &   0        \\
    \vdots           &  \ddots  &   \ddots     &    0            &   0       \\
    0                &  \dots   &    0         &    1            &   0
 \end{array} \right].
\end{align} We can now replace the scaled adjacencies $\frac{1}{\gamma}A_s$ and $\frac{1}{\gamma}A_p$ in \eqref{eq:L} and with $B$ given per Assumption~\ref{assump:b-equal-i}, we can solve for the corresponding Gramians $G_c^s$ and $G_c^p$ as follows. 

\textbf{Directed star} networks have a controllability Gramian $G_c^{s}$ whose entries are given by $[G_c^{s}]_{11} = [B]_{11}$, $[G_c^{s}]_{ii} = \gamma^{-2}[G_c^{s}]_{11} + [B]_{ii}, \forall i > 1$, $[G_c^{s}]_{j1} = [G_c^{s}]_{1j} = [B]_{1j} = 0$, $\forall j > 1$, and  $[G_c^{s}]_{ij} = \gamma^{-2}[G_c^{s}]_{11}$, $\forall j > 1, j \neq i$. In particular, all entries on the first row of the Gramian are zero except for the $1,1$ entry which is equal to $[B]_{11}$. Hence, in order for a star topology to be controllable the designer should always select the first (central) node as a leader. Further calculation of the eigenvalues indicate that we alway need to select all but one peripheral node of the star network to order to have a full rank Gramian or a controllable network, $\lambda_N(G_c^s)  > 0$. Therefore, \emph{ there is no fraction $f<1$ of nodes that can be chosen to ensure controllability of a star network, as $N \to \infty$}. This behavior is in sharp contrast with the directed path and cycle topologies analyzed next. In the latter cases, although no finite collection of leaders can ensure controllability as $N\to \infty$, the designer can still select an asymptotically non-vanishing fraction of nodes as leaders and obtain a controllable ring or path network for all $N$ and as $N \to \infty$.

\textbf{Directed path} networks have a diagonal controllability Gramian $G_c^{p}$ whose diagonal entries are give by $[G^p_c]_{ii} = \sum_{k=1}^{i}\gamma^{2(k-i)}[B]_{kk}, \forall i$. The eigenvalues of $G_c^{p}$ are the same as its diagonal entries, and the designer would again need to select the first (root) node if the system is to be controllable. However, with just the root node as the leader $\lambda_N(G_c^p) = \gamma^{2(1-N)} \to 0$ as $N \to \infty$ so that injecting the input just at the first node cannot ensure the controllability of an infinite integrator chain with the increasing length. Indeed, with any finite collection of leaders it follows that the distance to the leader nodes grows for the nodes further through the chain and the minimum eigenvalue of the Gramian would approach zero geometrically fast as $N \to \infty$. On the other hand, by selecting a non-zero fraction $f$ of nodes as leaders and spreading them uniformly throughout the chain one can ensure a  distance of at most $1/f$ to the closest leader for every node in the chain and the above calculation of the Gramian would then imply a lower bound of $\lambda_N(G_c^p)\geq \gamma^{-2/f}$, which holds even as $N \to \infty$. \emph{By selecting a non-zero fraction of nodes and spreading them uniformly throughout the network, the designer can ensure the controllability of an infinite integrator chain.} The numerical experiments in what follows indicate that this observation applies also to the case of networks with directed ring topology.

\textbf{Directed ring} networks demonstrate an asymptotic behavior that resembles that of the path networks as $N \to \infty$. Here, we investigate the effect of the fraction of leaders on the least eigenvalue of the Gramian in a ring network of $800$ nodes. To this end, we first divide the nodes into consecutive blocks of a fixed length and with varying number of leader nodes at each block. We next consider the effect of varying the block length by fixing only one leader at each block and increasing the block length. The two experiments in Fig.~\ref{fig:ratio-of-leaders-ring} indicate although the worst case least control effort decreases with the increasing fraction of leader nodes, when a single leader is fixed at each block better control can be achieved with a smaller fraction of leaders, since the leaders are better spread throughout the the network. Indeed, in the extreme case where all the leaders are clustered together then no fraction $f < 1 $ of leaders can ensure controllability as $N$ increases. This can be attributed to the fact that even though  the number of leader nodes increases with the network size, when all the leaders are clustered together and not spread through the network there will always be some nodes in the network that get arbitrarily far from all the leader as the network size increases.

\begin{figure}[h*]
\centering
\includegraphics[width=0.8\linewidth]{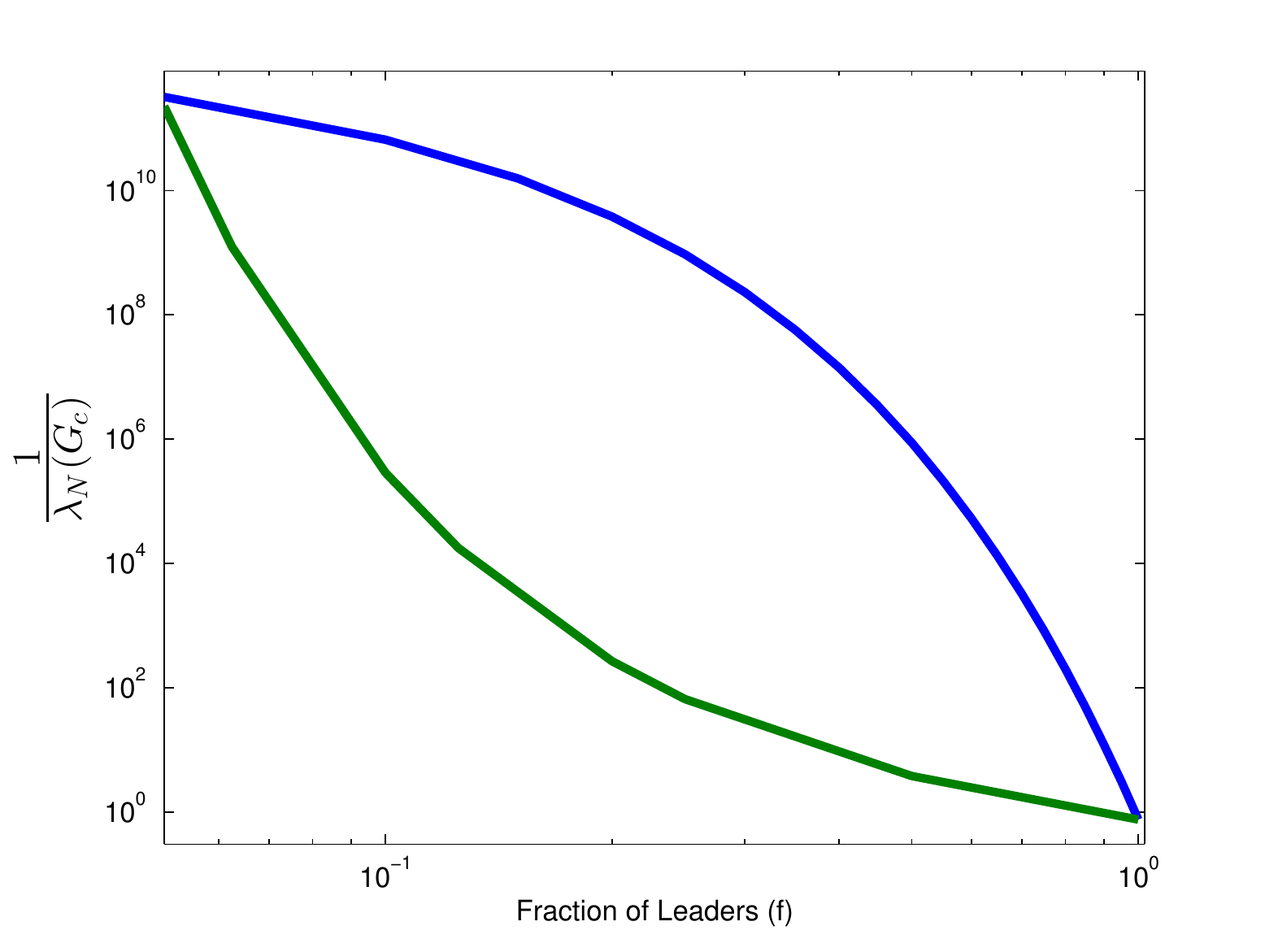}
\caption{Effect of the fraction and spread of leaders for an $800$-node directed cycle network: putting a single leader in each block and varying the block lengths over the first $11$ divisors of $800$ for the green curve; and fixing block length at $20$ and increasing the number of nodes at each block for blue curve.}
\label{fig:ratio-of-leaders-ring}
\end{figure}


\section{Conclusions}\label{sec:conc}

In this paper, we investigated the controllability of a linear single integrator network as the number of nodes increases. We first injected input signals at every node and required the controllability Gramian to remain well-conditioned even as the network size increases. Accordingly, with a proper normalization that is uniform in the size of the network, the Gramian condition number for graphs with a bounded maximum degrees was shown to remain bounded, uniformly in the size. The results provide theoretical insights on the challenges of controllability for large networks in general, and highlights the role of bounded degrees in particular. Furthermore, we proffered bounds on the condition number of the controllability Gramian, which in the cases of cycle, path or star topologies were expressible in terms of the network size and could guarantee numerical stability with the increasing dimension. We next shifted our attention to the question of choice and number of leader nodes for large networks, and showed that while a star topology can never be made controllable for all $N$ by selecting any fixed fraction $f<1$ of nodes as leaders; in the cases of path and ring networks, by selecting a non-zero fraction of nodes as leaders and having them spread across the network such that no nodes gets arbitrarily far from all leaders, the designer can ensure that the minimum eigenvalue of the Gramian is bounded away from zero even as the network size increases. This distinction between the star topology and path or rings with respect to the required asymptotic fraction of leaders for controllability with the increasing size, further highlights the challenges imposed by the high-degree nodes on the controllability of large networks.

\end{document}